\documentclass[conference,letterpaper]{IEEEtran}

\usepackage[utf8]{inputenc} 
\usepackage[T1]{fontenc}
\usepackage{url}
\usepackage{ifthen}
\usepackage{cite}
\usepackage[cmex10]{amsmath} 
\usepackage{cite}
\usepackage{latexsym}
\usepackage{amsthm}
\usepackage{caption}
\usepackage{amssymb}
\usepackage{mdframed}
\usepackage{graphicx}
\usepackage[colorlinks=true, allcolors=blue] {hyperref}
\usepackage{url}
\usepackage{comment}
\usepackage{algorithm, algorithmic}
\usepackage{mathtools}
\makeatletter
\def\th@theoremstyle{\thm@headfont{\bfseries}}
\makeatother
\theoremstyle{theoremstyle}  
\newtheorem{theorem}{Theorem} 
\newtheorem{lemma}[theorem]{Lemma} 
\newtheorem{definition}[theorem]{Definition}

\newtheorem{proposition}[theorem]{Proposition}
\newtheorem{remark}[theorem]{Remark}

\begin{document}
\title{Channel Resolvability Using\\ Multiplicative Weight Update Algorithm}

\author{\IEEEauthorblockN{Koki Takahashi and Shun Watanabe}
\IEEEauthorblockA{Department of Computer Sciences,
Tokyo University of Agriculture and Technology, Japan, \\
E-mail:s251249z@st.go.tuat.ac.jp; shunwata@cc.tuat.ac.jp}}

\maketitle


\begin{abstract}
We study channel resolvability, which is a key problem to proving a strong converse of identification via channel. The literature has only proved channel resolvability using random coding. We prove channel resolvability using the multiplicative weight update algorithm, the first approach to channel resolvability using non-random coding.

\end{abstract}

\section{Introduction}
Shannon introduced random coding to prove channel coding \cite{Shannon48}. Random coding is the most widely used approach in information theory. 
Feinstein also proved channel coding by introducing maximal code construction \cite{Feinstein59}. Feinstein's non-random coding approach has intensified our understanding of channel coding as seen in \cite{Csiszar-Korner11,Han03_Information-Spectrum,Ash65,Blackwell_Breiman_Thomasian59,Ogawa_Nagaoka07,Winter99}. 

Ahlswede and Dueck proposed identification via channel \cite{Alswede_Dueck89}, which ensures the message can be identified whether coming or not, but cannot be restored. 
In the transmission problem,  exponentially many messages of the block length can be transmitted
; on the other hand, in the identification problem, doubly exponentially many messages can be identified. 

Han and Verd{\'u} proved the strong converse of identification \cite{Han_Verdu92}, and summarized the concept of proof problem as channel resolvability \cite{Han_Verdu93}. Ahlswede also proved the strong converse of the identification via channel by essentially providing an alternative proof of channel resolvability \cite{Ahlswede06}.
In addition to the converse proof technique for the identification, channel resolvability also plays an important role in the context of the wiretap channel \cite{Hayashi06, Bloch_Laneman_13,Csiszár96}.

The problem of identification via channel and channel resolvability have been studied extensively \cite{Bracher_Lapidoth17, Boche_Schaefer20, Burnashev00, Hayashi06, Oohama13, Steinberg98, Steinberg_Merhav01, Yamamoto_Ueda15,Torres-Figueroa_Ferrara_Deppe_Boche23,Rosenberger_Pereg_Deppe23_compound, Wiese_Labidi_Deppe_Boche23, Rosenberger_Deppe_Pereg23_quantum,Kuzuoka18,Yagi_Han17}.
For point-to-point channel, the identification capacity is well understood. 
However, for multi-user networks, many identification problems are unsolved; for instance, the identification capacity region of the broadcast channel (with maximal error criterion) is unknown; see \cite[Section 2.4]{Bracher16}.
Thus we need to deepen our understanding of identification; see also a recent survey \cite{Brüche_Mross_Zhao_Labidi_Deppe_Jorswieck24}. 

%

As a first step toward that direction, we prove channel resolvability with non-random coding. As far as we know, there is no paper about channel resolvability with non-random coding, although various analysis techniques have been developed \cite{Han_Verdu93,Ahlswede06,Hayashi06,Oohama13,Cuff13,Ahlswede_Winter02,Li_Anantharam21,Yagli_Cuff19, Bastani_Telatar_Merhav17, Yu_Tan19,Liu_Cuff_Verdu17}. As we mentioned above, channel coding has two proof approaches; random coding and non-random coding. We believe that non-random coding approach to channel resolvability will deepen our understanding of  identification via channel.    

In this paper, we use the multiplicative weight update (MWU) algorithm to prove channel resolvability. The MWU algorithm is a widely used algorithm in game theory and learning theory, and it finds a certain minimax value of two-player game \cite{Arora_Hazan_Kale12}. By interpreting channel resolvability as two-player game and introducing a cost judiciously, we prove channel resolvability using non-random coding. 
The key idea is to use the MWU algorithm instead of the Chernoff bound used in Alswede's proof \cite{Ahlswede06}. Our proof is inspired by Kale's result, which proved hypergraph  covering problem using the MWU algorithm \cite{Kale07}. However, in contrast to hypergraph covering problem, channel resolvability is "soft covering" problem \cite{Cuff13}, and there are technical difficulties that are unique to channel resolvability. 

\textit{Notation:} Throughout the paper, random variables (e.g. $X$) and their realizations (e.g. $x$) are denoted by capital and lower letters, respectively. Finite alphabet sets (e.g. ${\cal X}$)  are denoted by calligraphic letters. For a finite set ${\cal X}$, the complement and the cardinality are denoted by ${\cal X}^c$, $|{\cal X}|$ respectively. The set of all distributions on ${\cal X}$ is denoted by ${\cal P}({\cal X})$.

\section{Problem Formulation of channel resolvability}


Consider an input distribution $P\in {\cal P(X)},$ and a channel $W$ from input alphabet ${\cal X}$ to output alphabet ${\cal Y}$. The goal is to simulate the output distribution
$$W_{P} := \sum_{{x}\in{\cal X}}P({x}) W_{x}.$$ 
Let ${\cal C} = \{{x}_1, {x}_2, ...,x_{L}\}$ be a codebook of size $L$ and $W_{{\cal C}}$ be the output distribution where inputs are codewords which are converted from uniform random numbers with the codebook ${\cal C}$, i.e.,
$$W_{{\cal C}}:=\frac{1}{L}\sum_{l=1}^L W_{x_l}.$$
A simulation error is measured by the total variation distance
between the output distribution $W_{P}$ and 
the output distribution $W_{{\cal C}}$ of the codebook:
$$d_\mathrm{var} (W_{P}, W_{{\cal C}}) := \frac{1}{2}\sum_{y\in{\cal Y}}|W_{P}(y) -W_{{\cal C}}(y)|.$$
The problem "channel resolvability" is how to construct a codebook ${\cal C}$ such that for arbitrary small $\eta$, the total variation distance $d_\mathrm{var}(W_{P}, W_{{\cal C}})$ is smaller than $\eta,$ as small size $L$ as possible.  


For the asymptotic setting of discrete memoryless channel and the worst input distribution, the optimal rate for channel resolvability is given by the Shannon capacity $C(W)$ \cite{Han_Verdu93}. Our goal of this paper is to achieve this result using non-random coding.


\section{Multiplicative weight update algorithm}

In this section, we review the MWU algorithm in notations that are compatible with
channel resolvability.
We consider an $L$ rounds game between ${\cal X}$-player and ${\cal Y}$-player;
${\cal X}$-player seeks to maximize the cost while ${\cal Y}$-player seeks to minimize
the cost. For ${\cal X}$-player's strategy $x\in {\cal X}$ and ${\cal Y}$-player's strategy
$y\in {\cal Y}$, cost $0 \le m(x,y) \le 1$ is incurred. 
In each round $1 \le l \le L$, ${\cal Y}$-player select a (mixed) strategy, i.e., 
a distribution $f(l,\cdot)$ on ${\cal Y}$; then ${\cal X}$-player responds with a 
strategy $x_l \in {\cal X}$. Our goal is to construct a sequence of mixed strategies
$f(1,\cdot), \ldots, f(L,\cdot)$ such that the total cost
\begin{align*}
\sum_{l=1}^L \bigg(\sum_{y^\prime \in {\cal Y}} f(l,y^\prime) m(x_l,y^\prime) \bigg)
\end{align*}
is not too much more than the cost of the ${\cal Y}$-player's best strategy in hindsight, i.e.,
\begin{align*}
\min_{y\in {\cal Y}} \sum_{l=1}^L m(x_l,y).
\end{align*}
The MWU algorithm is an iterative algorithm that obtains approximately the best strategies. The MWU algorithm has a few variations, and we use the so-called Hedge algorithm \cite{Arora_Hazan_Kale12}. 
The MWU algorithm reviewed in Algorithm \ref{algo:label} is a classical version of the one presented in \cite{Kale07}, and its performance is evaluated in the following lemma. For readers' convenience, we provide a proof of
the lemma in Appendix \ref{Proof_Lemma2}
; see also  \cite[Theorem 2.3]{Arora_Hazan_Kale12}.

\begin{algorithm}[t]
 \caption{Multiplicative weight update algorithm}
 \label{algo:label}
 \begin{algorithmic}
    \STATE Fix $\varepsilon \in(0, \frac{1}{2}],$ and initialize weights $w(1,y)=1$ for $y\in {\cal Y}$.
    \\For each step $l=1,2,..., L$:
        \\1. Calculate ${\cal Y}$-player's mixed strategy as follows:
         $$f(l,y) = \frac{w(l,y)}{\sum_{y^\prime\in{\cal Y}}w(l,y^\prime)},~y \in {\cal Y}.$$\\
        2. ${\cal X}$-player selects a pure strategy $x_l \in {\cal X}$.\\
        3. For the cost $(m(x_l,y):y\in {\cal Y})$, weights are updated via:
        $$w(l+1, y) = \exp\bigg(-\varepsilon \sum_{l' = 1}^{l} m(x_{l'},y)\bigg),~y\in{\cal Y}$$
 \end{algorithmic}
\end{algorithm}

\begin{lemma}(MWU algorithm performance)  
\label{lem:MWU_Algorithm}
Fix an arbitrary $\varepsilon\in (0, \frac{1}{2}]$.
The mixed strategies $f(1,\cdot),\ldots,f(L,\cdot)$ of ${\cal Y}$-player's strategy in Algorithm \ref{algo:label} satisfies for every $y\in {\cal Y}$,
\begin{align*}
        (1-\varepsilon) \sum_{l=1}^L \bigg(\sum_{y^\prime\in{\cal Y}} f(l,y^\prime)m(x_l,y^\prime)\bigg) 
        \le 
        \sum_{l=1}^Lm(x_l,y) + \frac{\ln{|{\cal Y}|}}{\varepsilon}.
\end{align*}
\end{lemma}


\section{Single-shot bound}

In this section, we derive a single-shot bound for channel resolvability using the MWU algorithm. 
For the later convenience, we consider truncated measures known as smoothing, to apply the single-shot bound for the asymptotic setting; e.g. see \cite{Watanabe21}. 
\begin{definition}(Truncated measures)
    \label{def:channel}
    For an arbitrary set ${\cal S} \subseteq {\cal X \times Y}$, and a given channel $W$,  
    we define truncated measure as
    $$W_x^{\cal S}(y) = W_x(y) \mathbf{1}[(x,y)\in {\cal S}],$$
    and truncated measure with input distribution $P$ as
    $$W_P^{\cal S}(y) = \sum_{x\in{\cal X}}P(x)W_x(y) \mathbf{1}[(x,y)\in {\cal S}].$$
    We define $W_{\cal C}^{\cal S}$ similarly. 
    For the complement ${\cal S}^c$, we define $W_x^{{\cal S}^c}, W_P^{{\cal S}^c}, W_{\cal C}^{{\cal S}^c}$ similarly.
\end{definition}
We also consider the support set 
$$\mathsf{s}(W_P^{\cal S}) := \{y : W_P^{\cal S}(y)>0\},$$
to define a cost properly.
The following choice of cost enables us to apply the MWU algorithm for channel resolvability.
\begin{definition}(Cost)
    \label{def:cost}
    For an input distribution $P$, a given channel $W$, and an arbitrary set ${\cal S}$, 
    we define cost as 
    $$m(x,y) = \begin{dcases} \frac{W_x^{\cal S}(y)}{W_P^{\cal S}(y)}e^{-D_\mathrm{max}} & y \in \mathsf{s}(W_P^{\cal S})
        \\  \mathrm{undefined} & y\notin \mathsf{s}(W_P^{\cal S}) 
    \end{dcases},$$
    where 
    $$D_\mathrm{max} = \max_{(x,y)\in {\cal S} : \atop y \in \mathsf{s}(W_P^{\cal S})}\ln{\frac{W_x^{\cal S}(y)}{W_P^{\cal S}(y)}}.$$
\end{definition}
$D_\mathrm{max}$ ensures that the cost $m(x,y)$ is smaller than or equals to $1$ for all strategy $(x,y)\in {\cal X}\times \mathsf{s}(W_P^{\cal S})$. We run Algorithm \ref{algo:label} for $\mathsf{s}(W_P^{\cal S})$ instead of ${\cal Y}$; then, ${\cal Y}$-player's mixed strategy $f(l,\cdot)$ is a distribution on $\mathsf{s}(W_P^{\cal S})$.
Finally, we define ${\cal X}$-player's strategy at each round $l$. 
\begin{definition}(${\cal X}$-player's strategy) 
    \label{def:X_strategy}
    At each round $1\leq l\leq L$, ${\cal Y}$-player selects a mixed strategy $f(l,y)$ for $y\in\mathsf{s}(W_P^{\cal S})$. For cost $m(x,y)$, ${\cal X}$-player selects $x_l\in{\cal X}$ such that
    $$x_l = \underset{x}{\mathrm{argmax}} \sum_{y\in\mathsf{s}(W_P^{\cal S})} f(l,y)m(x,y). $$
\end{definition}

Note that such ${\cal X}$-player's strategy $x_l$ satisfies Lemma \ref{lem:X_strategy}.
\begin{lemma}(${\cal X}$-player's strategy performance)
    \label{lem:X_strategy}
    At each round $l$, after ${\cal Y}$-player's mixed strategy $f(l,y)$ for $y\in \mathsf{s}(W_P^{\cal S})$, ${\cal X}$-player's strategy $x_l$ satisfies
    \begin{equation}
        \label{eq:sec4:X_strategy}
        \sum_{y\in \mathsf{s}(W_P^{\cal S})} f(l,y) m(x_l, y) \geq e^{-D_\mathrm{max}}.
    \end{equation}
\end{lemma}
\begin{proof}
We get 
\begin{equation*}
    \begin{split}
        \max_x \sum_{y\in \mathsf{s}(W_P^{\cal S})}f(l,y) \frac{W_x^{\cal S}(y)}{W_P^{\cal S}(y)} &\geq \mathbb{E}_{X \sim P}\bigg[\sum_{y\in\mathsf{s}(W_P^{\cal S})}f(l,y) \frac{W_X^{\cal S}(y)}{W_P^{\cal S}(y)}\bigg]
        \\&= \sum_{y\in\mathsf{s}(W_P^{\cal S})}f(l,y) \frac{W_P^{\cal S}(y)}{W_P^{\cal S}(y)}
        \\&= \sum_{y\in \mathsf{s}(W_P^{\cal S})}f(l,y) \\&= 1.
    \end{split}
\end{equation*}
Multiplying $e^{-D_{\mathrm{max}}}$ to both sides, we get equation \eqref{eq:sec4:X_strategy} for $y\in \mathsf{s}(W_P^{\cal S})$. 
\end{proof}

For our choice of cost and ${\cal X}$-player's strategy, we get the following performance of channel resolvability code ${\cal C}=\{x_1, x_2, \ldots,x_L\}$.

\begin{lemma}(Channel resolvability performance)
\label{lem:channelResl_performance}
For an input distribution $P$, a given channel $W$, an arbitrary set ${\cal S}$,  and arbitrary small $\varepsilon\in(0,\frac{1}{2}]$, the codebook ${\cal C}$ obtained by the MWU algorithm satisfies
\begin{equation}\label{eq:sec4:channel_Resolvperfor}
\forall y\in {\cal Y} \qquad W_P^{\cal S}(y) -W_{\cal C}^{\cal S}(y)\leq 2\varepsilon W_P^{\cal S}(y),  \end{equation}
as long as the size $L$ satisfies that 
\begin{equation}\label{eq:sec4:assumption}
    L \geq \frac{e^{D_\mathrm{max}}\ln{|{\cal Y}|}}{\varepsilon^2}.
\end{equation}
\end{lemma}
\begin{proof}
    First, we derive the equation 
    \begin{equation} \label{eq:sec4:2}
        \forall y \in {\cal Y} \qquad  W_{{\cal C}}^{\cal S} (y) \geq \bigg(1-\varepsilon - \frac{e^{D_\mathrm{max}} \ln{|{\cal Y}|}}{L\varepsilon}\bigg)W_P^{\cal S}(y),
    \end{equation}
    with two cases,  $y\notin \mathsf{s}(W_P^{\cal S})$ and $y\in\mathsf{s}(W_P^{\cal S})$.

    Case $y \notin \mathsf{s}(W_P^{\cal S})$:
    For all $y \notin \mathsf{s}(W_P^{\cal S})$, $W_P^{\cal S}(y)$ is 0 and $W_{\cal C}^{\cal S}(y)$ is not negative. Thus, equation \eqref{eq:sec4:2} holds.
    
    Case $y\in \mathsf{s}(W_P^{\cal S})$: 
    Applying Lemma \ref{lem:MWU_Algorithm} for 
    Definition \ref{def:X_strategy}, Lemma \ref{lem:X_strategy} implies that
    \begin{equation*}
        \forall y \in \mathsf{s}(W_P^{\cal S}) \qquad \sum_{l=1}^Lm(x_l,y)\geq (1-\varepsilon)L e^{-D_\mathrm{max}} - \frac{\ln{|\mathsf{s}(W_P^{\cal S})|}}{\varepsilon}.
    \end{equation*}
    From Definition \ref{def:cost}, for $y\in\mathsf{s}(W_P^{\cal S})$, the left side of equation is
    $$ \sum_{l=1}^Lm(x_l,y) = \sum_{l=1}^L \frac{W_{x_l}^{\cal S}(y)}{W_P^{\cal S}(y)}e^{-D_\mathrm{max}} = L\frac{W_{\cal C}^{\cal S}(y)}{W_P^{\cal S}(y)}e^{-D_\mathrm{max}} .$$
    Then multiplying $\frac{W_P^{\cal S}(y)e^{D_\mathrm{max}}}{L}$ to both sides, we get
    $$\forall y \in \mathsf{s}(W_P^{\cal S}) \qquad  W_{{\cal C}}^{\cal S} (y) \geq \bigg(1-\varepsilon - \frac{e^{D_\mathrm{max}} \ln{|\mathsf{s}(W_P^{\cal S})|}}{L\varepsilon}\bigg)W_P^{\cal S}(y).$$
    Because support set $\mathsf{s}(W_P^{\cal S})$ is included in set ${\cal Y}$, we get $|\mathsf{s}(W_P^{\cal S})| \leq |{\cal Y}|.$
    Thus, equation \eqref{eq:sec4:2} holds.
    
    Adjusting the size $L$ in equation \eqref{eq:sec4:2}, we get
    \begin{equation}\label{eq:sec4:3}
        \forall y \in {\cal Y}  \qquad W_{\cal C}^{\cal S}(y)\geq (1-2\varepsilon ) W_P^{\cal S}(y).
    \end{equation}
    as long as the size $L$ satisfies assumption \eqref{eq:sec4:assumption}.
    After rearranging equation \eqref{eq:sec4:3}, we get equation \eqref{eq:sec4:channel_Resolvperfor}
\end{proof}

From Lemma \ref{lem:channelResl_performance}, we derive a single-shot bound for channel resolvability using the MWU algorithm.

\begin{theorem}(Single-shot bound for channel resolvability) 
    \label{the:Single-shot}
    For an input distribution $P$, a given channel $W$, an arbitrary set ${\cal S}$, and arbitrary small $\varepsilon\in(0,\frac{1}{2}]$, the codebook ${\cal C}$ obtained by the MWU algorithm satisfies
    $$d_{\mathrm{var}}(W_P, W_{\cal C}) \leq 2\varepsilon + \sum_{(x,y)\in{\cal S}^c}P(x)W_x(y),$$
    as long as the size $L$ satisfies that 
    $$L \geq  \frac{e^{D_\mathrm{max}} \ln{|{\cal Y}|} }{\varepsilon^2}.$$
\end{theorem}

\begin{proof} 
    We measure the total variation distance and get
    \begin{align}
    \label{eq:sec4:varDist}
        d_\mathrm{var}(W_P,W_{\cal C}) 
        &= \frac{1}{2}\sum_{y\in{\cal Y}}|W_P(y)-W_{\cal C}(y)| \notag \\
        &= \sum_{y\in {\cal Y} : \atop W_P(y)\geq W_{\cal C}(y)}\Big(W_P(y) - W_{\cal C}(y) \Big) \notag\\
        & = \sum_{y\in {\cal Y} : \atop W_P(y)\geq W_{\cal C}(y)} \Big(W_P^{\cal S}(y) + W_P^{{\cal S}^c}(y)\Big)  \notag\\
        & \qquad - \sum_{y\in {\cal Y} : \atop W_P(y)\geq W_{\cal C}(y)} \Big(W_{\cal C}^{\cal S}(y) + W_{\cal C}^{{\cal S}^c}(y)\Big) \notag\\
         & = \sum_{y\in {\cal Y} : \atop W_P(y)\geq W_{\cal C}(y)} \Big(W_P^{\cal S}(y) - W_{\cal C}^{{\cal S}}(y)\Big)  \notag \\ 
         &\qquad  + \sum_{y\in {\cal Y} : \atop W_P(y)\geq W_{\cal C}(y)} \Big(W_P^{{\cal S}^c}(y) - W_{\cal C}^{{{\cal S}}^c}(y)\Big).
    \end{align}
    Lemma \ref{lem:channelResl_performance} implies that 
    \begin{align}\label{eq:sec4:chaResol_performance}
        \sum_{y\in {\cal Y} : \atop W_P(y)\geq W_{\cal C}(y)} \Big(W_P^{\cal S}(y) - W_{\cal C}^{{\cal S}}(y)\Big)  &\leq 2\varepsilon \sum_{y\in {\cal Y} : \atop W_P(y)\geq W_{\cal C}(y)}W_P^{\cal S}(y) \notag
        \\& \leq 2\varepsilon.
    \end{align}
    By combining equation \eqref{eq:sec4:varDist} and \eqref{eq:sec4:chaResol_performance}, we get 
   \begin{align*}
        d_\mathrm{var}(W_P,W_{\cal C}) & \leq 2\varepsilon + \sum_{y\in {\cal Y} : \atop W_P(y)\geq W_{\cal C}(y)} \Big(W_P^{{\cal S}^c}(y) - W_{\cal C}^{{{\cal S}}^c}(y)\Big) 
        \\&\leq 2\varepsilon + \sum_{y\in {\cal Y} : \atop W_P(y)\geq W_{\cal C}(y)} W_P^{{\cal S}^c}(y)
        \\&\leq  2\varepsilon + \sum_{y\in{\cal Y}} W_P^{{\cal S}^c}(y).
    \end{align*}
    Note that
    $\sum_{y\in{\cal Y}}W_P^{{\cal S}^c}(y) = \sum_{(x,y)\in{\cal S}^c}P(x)W_x(y).$
    Thus, we get Theorem \ref{the:Single-shot} as long as the size $L$ satisfies assumption \eqref{eq:sec4:assumption}.
\end{proof}



\section{Asymptotic result for fixed type}
In this section, we consider channel resolvability for block length $n$ 
and use type method; see \cite[Section 2]{Csiszar-Korner11}. 
We fix a type $T$ on ${\cal X}$. Let ${\cal X}_T^n$ be a set of all sequences $x^n$ in ${\cal X}^n$ that have type $T$.
We consider an input distribution $P_n^T$ on ${\cal X}_T^n$. 
We simulate the output distribution $W_{P_n^T}$ for an input distribution $P_n^T$ using the MWU algorithm. We define the set ${\cal A}$ on the strong conditional typical set of the space ${\cal (X}_T^n {\cal ,Y}^n)$. 
\begin{definition}(Strong conditional typical set)
\label{def:condi_typical}
    For block length $n$, a type $T$, a given i.i.d channel $W$, and small error threshold $\alpha$, we define the strong conditional typical set 
    \begin{equation*}
        \begin{split}
            {\cal T}_{W,\alpha}^n(x^n) 
            :=& \bigg\{y^n: \bigg|  \frac{1}{n}N(a,b|x^n,y^n)- T(a)W_a(b) \bigg| < \alpha \\&\qquad  \forall a\in{\cal X}, b\in{\cal Y}\bigg\},     
        \end{split}
    \end{equation*}
    where $N(a,b|x^n,y^n)$ outputs the number of occurrences of set $(a,b)$ in $(x^n,y^n)$.
    We also define set
    $${\cal A} := \{(x^n, y^n): x^n \in {\cal X}_T^n,\; y^n \in {\cal T}_{W,\alpha}^n(x^n)\}.$$
\end{definition}

Note that if $y^n$ is in the strong conditional typical set ${\cal T}_{W,\alpha}^n$,$\;$ $y^n$ satisfies the property of the weak conditional typical sequence such that 
\begin{align}\label{eq:sec5:weak_conditional}
     \bigg|  -\frac{1}{n}\ln W_{x^n}(y^n)- H(W|T) \bigg| < \alpha' 
     \quad \forall x^n\in{\cal X}_T^n, 
\end{align}
where $\alpha' = c_W \alpha$ for constant $c_W$  \cite[Problem 2.5]{Csiszar-Korner11}. We define the set ${\cal B}$ on the strong typical set of the output space ${\cal Y}^n$. 

\begin{definition}(Strong typical set)
\label{def:typical}For block length $n$, a type $T$, a given i.i.d channel $W$, and small error threshold $\beta = |{\cal X}|\alpha$, we define the strong typical set 
\begin{equation*}
    \begin{split}
         {\cal T}_{W_{T}, \beta}^n:=
          \bigg\{y^n: \bigg| \frac{1}{n}N(b|y^n)-W_T(b) \bigg|\leq \beta \quad \forall b\in {\cal Y} \bigg\}
    \end{split}
\end{equation*}
where $N(b|y^n)$ outputs the number of occurrences of $b$ in $y^n$.
For an input distribution $P_n^T$, we also define set 
    \begin{equation*}
        {\cal B}:= \bigg\{y^n \in {\cal T}_{W_{T}, \beta}^n: W_{P_n^T}^{\cal A}(y^n) \geq \frac{\tau}{|{\cal T}_{W_{T},\beta}^n|} \bigg\},
    \end{equation*}
    where $\tau$ is a small positive value and 
    $$W_{P_n^T}^{\cal A}(y^n) = \sum_{x^n\in{\cal X}^n}P_n^T(x^n)W_{x^n}(y^n)\mathbf{1}[(x^n, y^n)\in \mathcal{A}].$$
\end{definition}
Note that the number of strong typical sequences satisfies
\begin{equation}\label{eq:sec5:number_typical}
    \bigg|\frac{1}{n}\ln |{\cal T}_{W_T,\beta}^n| - H(W_T) \bigg|\leq \beta',
\end{equation}
 for arbitrarily small positive $\beta'$ and sufficiently large $n$ \cite[Lemma 2.13]{Csiszar-Korner11}.
Using these sets ${\cal A}$ and ${\cal B}$, we prove channel resolvability for a fixed type. 

\begin{theorem}(Channel resolvability for fixed type)
\label{the:chaResol_fixedType}
     For a fixed type $T$, an input distribution $P_n^T$ on ${\cal X}_T^n$, a given i.i.d. channel $W$, arbitrarily small  $\varepsilon \in(0,\frac{1}{2}]$, small positive $\tau, \delta, \alpha' ,\beta'>0$, and sufficiently large $n$, the codebook ${\cal C}_n^T$ obtained by the MWU algorithm satisfies
     \begin{equation*}\label{eq:sec5:theo_fixedType}
         d_\mathrm{var}(W_{P_n^T}, W_{{\cal C}_n^T})\leq 2\varepsilon + \tau + \delta,
     \end{equation*}
    as long as the size $L_n$ satisfies that 
    \begin{equation*}\label{eq:sec5:assumption}
        L_n \geq \frac{\exp{\Big(n(I(T,W) + \alpha'+\beta') -\ln{\tau}\Big)}\; \ln{|{\cal Y}|^n} }{\varepsilon^2},
    \end{equation*}
    where $I(T,W)$ is mutual information of type $T$ and channel $W$.
\end{theorem}
\begin{proof}
    Fix sufficiently small $\alpha$ and let $\beta=|{\cal X}|\alpha$ so that equation \eqref{eq:sec5:weak_conditional} and \eqref{eq:sec5:number_typical} are satisfied for the given $\alpha'$ and $\beta'$.
    We define a set 
    $${\cal S} := {\cal A} \cap ({\cal X}_T^n \times  {\cal B}). $$
    To apply Theorem \ref{the:Single-shot} for the set ${\cal S}$, we have to calculate $D_\mathrm{max}$, and $\sum_{(x^n,y^n)\in{\cal S}^c} P_n^T(x^n) W_{x^n}(y^n) .$ First, we focus on the bound of $W_{x_n}^\mathcal{S}(y^n)$ and $W_{P_n^T}^\mathcal{S}(y^n)$. For $(x^n, y^n) \in \mathcal{S}$, equation \eqref{eq:sec5:weak_conditional} leads that 
    \begin{equation}
        \label{eq:sec5:bound_W_x^S}
        \frac{1}{n}\ln{W_{x^n}^{\cal S}(y^n)} \leq -H(W|T)+\alpha',
    \end{equation} 
    Note that $$W_{P_n^T}^\mathcal{S}(y^n) = W_{P_n^T}^\mathcal{A}(y^n),$$ for $y^n\in\mathsf{s}(W_{P_n^T}^\mathcal{S})\subset \mathcal{B}$.
    Using the property of the set ${\cal B}$ in equation \eqref{eq:sec5:number_typical}, for $y^n\in\mathsf{s}(W_{P_n^T}^\mathcal{S})$, we obtain 
    \begin{align}
        \label{eq:sec5:bound_W_P^S}
            \frac{1}{n}\ln{W_{P_n^T}^{\cal S}(y^n)} &= \ln{W_{P_n^T}^{\cal A}(y^n)} \notag
        \\&\geq \frac{1}{n}\ln{\frac{\tau}{|{\cal T}_{W_{T},\beta}^n|}} \notag
        \\&\geq -H(W_{T}) - \beta' +\frac{1}{n}\ln{\tau}. 
    \end{align}
    Thus, from Definition \ref{def:cost}, we get
    \begin{align}\label{eq:sec5:D_max}
            D_\mathrm{max}&=\max_{(x^n,y^n)\in {\cal S} : \atop y^n \in \mathsf{s}\big(W_{P_n^T}^{\cal S}\big)} 
            \ln{\frac{W_{x^n}^{{\cal S}}(y^n)}{W_{P_n^T}^{{\cal S}}(y^n)}} \notag
            \\&= \max_{(x^n,y^n)\in {\cal S} : \atop y^n \in \mathsf{s}\big(W_{P_n^T}^{\cal S}\big)}
            \Big(\ln{W_{x^n}^{{\cal S}}(y^n)} - \ln{W_{P_n^T}^{{\cal S}}(y^n)}\Big)\notag
            \\&\leq n(-H(W|T)+\alpha') + n(H(W_T)+\beta')-\ln{\tau} \notag
            \\& = n(I(T,W) +\alpha'+\beta')-\ln{\tau},
    \end{align}
    where the inequality follows from equation \eqref{eq:sec5:bound_W_x^S}
    and \eqref{eq:sec5:bound_W_P^S}.
 
    Next, for ${\cal S}^c := ({\cal X}_T^n \times {\cal Y}^n) \backslash {\cal S}$, we get
    \begin{align}\label{eq:sec5:smoothing1}
            &\sum_{(x^n,y^n)\in{\cal S}^c} P_n^T(x^n) W_{x^n}(y^n) \notag
            \\&\qquad  \leq  \sum_{(x^n,y^n)\in ({\cal X}_T^n \times {\cal Y}^n)\backslash {\cal A}} P_n^T(x^n) W_{x^n}(y^n) \notag
            \\&\qquad\quad +  \sum_{(x^n,y^n) \in {\cal A} \cap ({\cal X}_T^n \times ({\cal Y}^n\backslash {\cal B}) )} P_n^T(x^n) W_{x^n}(y^n),
    \end{align}
    where the inequality follows from ${\cal S}^c=\{ ({\cal X}_T^n \times {\cal Y}^n)\backslash {\cal A} \} \cup  \{{\cal A} \cap ({\cal X}_T^n \times ({\cal Y}^n\backslash {\cal B}) )\}$
    with the union bound. 
    Note that the probability of non-typical set is small \cite[Lemma 2.12]{Csiszar-Korner11}, i.e., 
    for sufficient large $n$, we get
    \begin{align}\label{eq:sec5:smoothing2}
        \sum_{(x^n, y^n) \in ({\cal X}_T^n \times {\cal Y}^n)\backslash {\cal A}} P_n^T(x^n) W_{x^n}(y^n) \leq \delta.
    \end{align}
    Note also that the set ${\cal A} \cap \big({\cal X}_T^n\times ( {\cal Y}^n\backslash{{\cal T}_{W_T,\beta}^n})\big)$ is empty set, because  each sequence $y^n$ in ${\cal A}$ is included in ${\cal T}_{W_T,\beta}^n$ for $\beta = |{\cal X}|\alpha$ \cite[Lemma 2.10]{Csiszar-Korner11}. We get  
    \begin{align}\label{eq:sec5:smoothing3}
            &\sum_{(x^n,y^n)\in {\cal A} \cap ({\cal X}_T^n \times ({\cal Y}^n\backslash {\cal B}) )} P_n^T(x^n) W_{x^n}(y^n) \notag 
            \\& = \sum_{y^n \in  {{\cal T}_{W_T,\beta}^n}:\atop W_{P_n^T}^\mathcal{A}(y^n)<\frac{\tau}{|{\cal T}_{W_T,\beta}^n|}
            } \sum_{x^n\in{\cal X}_T^n} P_n^T(x^n) W_{x^n}(y^n) \mathbf{1}[(x^n, y^n)\in\mathcal{A}] \notag 
            \\&= \sum_{y^n \in  {{\cal T}_{W_T,\beta}^n}:\atop W_{P_n^T}^\mathcal{A}(y^n)<\frac{\tau}{|{\cal T}_{W_T,\beta}^n|}} W_{P_n^T}^\mathcal{A}(y^n)  
             \quad \leq  \tau.
    \end{align}
    Substituting equation \eqref{eq:sec5:smoothing2} and \eqref{eq:sec5:smoothing3} to equation \eqref{eq:sec5:smoothing1}, we get 
        $$\sum_{(x^n,y^n)\in{\cal S}^c} P_n^T(x^n) W_{x^n}(y^n) \leq \tau + \delta.$$
    Applying Theorem \ref{the:Single-shot} for equation \eqref{eq:sec5:D_max} and this equation, 
    we get Theorem \ref{the:chaResol_fixedType}.
\end{proof}


\section{Asymptotic result for general input distribution}

In this section, we consider channel resolvability for a general input distribution. Let ${\cal T}$ be the set of all types. For a given distribution $P_n$ on ${\cal X}^n$, we consider the type distribution
\begin{equation}
    \label{eq:definition_of_type_distribution}
    P_{{\cal T}}(T) := P_n({\cal X}_T^n)
\end{equation}
The type distribution distributes a general input distribution by each type $T$ as 
$$P_n (x^n) = \sum_{T\in{\cal T}}{P_{\cal T}(T)}{P_n^T(x^n)},$$
for each symbol $x^n \in {\cal X}^n$ where 
$$P_n^T(x^n) := \frac{P_n(x^n)}{P_n({\cal X}_T^n)} \mathbf{1}[x^n \in {\cal X}_T^n].$$
We first simulate the type distribution $P_{\cal T}$.  For readers' convenience, we provide a proof of the proposition in Appendix \ref{Proof:Proposition}; This simulation is known as source resolvability, which is already proved by deterministic code construction \cite{Han03_Information-Spectrum,Uematsu10, Nomura20}. However, we prove source resolvability using the MWU algorithm so that our channel resolvability code is constructed solely by the MWU algorithm.
\begin{proposition}(Source resolvability with type)
\label{propo:sourcdResolv}  For type distribution $P_{\cal T}$, arbitrary small $\varepsilon'\in (0,\frac{1}{2}],$ small positive $ \tau'>0$, and block length $n$, the codebook ${\cal C'} =\{T_1, T_2, ..., T_{L_n'} \}\subset {\cal T}$ obtained by the MWU algorithm makes a distribution
    $$R_{{\cal C}'}(T) = \frac{1}{L_n'} \sum_{l=1}^{L_n'} \mathbf{1}[T_l = T],$$ such that
    $$d_\mathrm{var}(P_{\cal T}, R_{{\cal C}'})\leq 2\varepsilon' + \tau',$$
    as long as the size $L_n'$ satisfies
    \begin{equation}
    \label{eq:L_n'}
        L_n' \geq \frac{(n+1)^{|{\cal X}|} |{\cal X}|\ln{(n+1)}}{\tau \varepsilon'^2}.
    \end{equation}
\end{proposition}

Select one of $L_n'$ types uniformly at random.  For the chosen type $T$, select a codeword uniformly at random from  $L_n$ codewords. 
In Proposition \ref{propo:sourcdResolv}, sufficiently large $L_n'$ ensures the outcome distribution approximates $P_{\cal T}$. Similarly, in Theorem \ref{the:chaResol_fixedType}, sufficiently large $L_n$ ensures the output distribution approximates the one induced by $P_n^T(x^n)$. These insights lead to channel resolvability for a general input.

\begin{theorem}(Channel resolvability for general input distribution)
    \label{the:Main_result}
     For a general input distribution $P_n$, a given i.i.d. channel $W,$ and small constants $\eta, \kappa$, the codebook ${\cal C}_n$ obtained by the MWU algorithm satisfies that
    $$d_\mathrm{var}(W_{P_n}, W_{{\cal C}_n}) \leq \eta,$$
    for sufficiently large $n$, as long as the size $L_n''$ satisfies that 
    $$L_n''\geq \exp (n(C(W)+ \kappa)).$$
\end{theorem}
\begin{proof}
    From Proposition \ref{propo:sourcdResolv}, we get type distribution $R_{{\cal C}'}$ 
    and from Theorem \ref{the:chaResol_fixedType}, we get distribution $W_{{\cal C}_n^T}$ for a fixed type $T$. 
    Thus, we get
    \begin{align}\label{eq:sec6:d_var_}
            &d_{\mathrm{var}}(W_{P_n},W_{{\cal C}_n}) \notag 
            \\& \quad = d_\mathrm{var}\bigg(\sum_{T \in {\cal T}}P_{\cal T}({T})W_{P_n^T}, \sum_{T \in {\cal T}}R_{{\cal C}'}({T})W_{{\cal C}_n^T}\bigg) \notag
            \\& \quad \leq d_\mathrm{var}\bigg(\sum_{T \in {\cal T}}P_{\cal T}({T})W_{P_n^T}, \sum_{T \in {\cal T}}P_{\cal T}({T})W_{{\cal C}_n^T}\bigg) \notag
            \\& \qquad+ d_\mathrm{var}\bigg(\sum_{T \in {\cal T}}P_{\cal T}({T})W_{{\cal C}_n^T}, \sum_{T \in {\cal T}}R_{{\cal C}'}({T})W_{{\cal C}_n^T}\bigg) \notag
            \\& \quad \leq \sum_{T \in {\cal T}} P_{\cal T}({T})d_\mathrm{var}\big(W_{P_n^T},W_{{\cal C}_n^T}\big)+d_\mathrm{var}(P_{\cal T}, R_{{\cal C}'}),
    \end{align}
where the first inequality follows from triangular inequality and the last inequality follows from data processing inequality. When we consider the worst case of type $T$ in Theorem \ref{the:chaResol_fixedType}, we get  $d_\mathrm{var}(W_{P_n^T},W_{{\cal C}_n^T}) \leq 2\varepsilon + \tau + \delta$ 
as long as the size $L_n$ satisfies that for channel capacity $C(W)=\max_T I(T,W)$,
\begin{equation*}
    \begin{split}
        L_n &\geq 
        \frac{\exp{\Big(n(C(W) + \alpha'+\beta') -\ln{\tau}\Big)}\; \ln{|{\cal Y}|^n} }{\varepsilon^2}.
    \end{split}
\end{equation*}
 Thus, we get
\begin{equation*}\label{eq:sec6:d_var}
    \begin{split}
        d_\mathrm{var}(W_{P_n},W_{{\cal C}_n}) \leq 2\varepsilon + \tau + \delta + 2\varepsilon' + \tau',
    \end{split}
\end{equation*}
as long as the size $L_n''$ satisfies that
\begin{equation*}\label{eq:sec6:L_n}
    \begin{split}
        L_n'' &\geq L_n \frac{(n+1)^{|{\cal X}|} |{\cal X}|\ln{(n+1)}}{\tau' \varepsilon'^2}.
    \end{split}
\end{equation*}
By setting $\varepsilon, \varepsilon', \tau, \tau', \delta, \alpha',\beta'$ to small values, 
for sufficiently large $n$,  
$\eta$ and $\kappa$ satisfies Theorem \ref{the:Main_result}.
\end{proof}

A condition where a input distribution $P_n$ is i.i.d. distribution alleviates Theorem \ref{the:Main_result}. For readers convenience, we provide a proof of following remark in appendix \ref{proof:Remark}.
\begin{remark}(Channel resolvability for i.i.d. input distribution)\label{Remark_}
    For i.i.d. input distribution $P_n=P^n$, in Theorem \ref{the:Main_result}, the size $L_n''$ only has to satisfy $L''_n\geq \exp(n (I(P, W)+\kappa)).$
\end{remark}


\section{Conclusion}
We prove channel resolvability using non-random coding. However, there are still two unresolved problems. The one is that we only consider discrete memoryless channels. We will consider general channels \cite{Han03_Information-Spectrum} using non-random coding. The other is to apply our result for quantum channel resolvability.


\section*{Acknowledgment}
This work was supported in part by the Japan Society for the Promotion of Science (JSPS) KAKENHI under Grant 23H00468 and 23K17455.

\newpage

\appendix
\subsection{Proof of Lemma 1}
\label{Proof_Lemma2}
    We shall calculate the upper bound and lower bound of the weight $w(L+1, y)$ summed for over all $y$. We calculate the upper bound first. Here, we have 
    \begin{equation}\label{eq:appe:1}
        \begin{split}
            w(L+1, y) &= \exp{\bigg(-\varepsilon \sum_{l=1}^L m(x_{l}, y) \bigg)}
            \\& = \exp \bigg(-\varepsilon \sum_{l=1}^{L-1} m(x_{l}, y) \bigg) \exp\big(-\varepsilon m(x_L, y)\big)
            \\&= w(L,y) \exp{\big(-\varepsilon m(x_L, y)\big)}.
        \end{split}
    \end{equation}
    Define $\varepsilon_1:= 1-e^{-\varepsilon},$ we get
    \begin{equation*}
        \exp{\big(-\varepsilon m(x_L, y)\big)} \leq 1-\varepsilon_1 m(x_L, y),
    \end{equation*}
    where the inequality follows from
    \begin{equation}\label{eq:appe:nontrivial_1}
        \exp(-\varepsilon z)\leq 1 - \varepsilon_1 z 
    \end{equation}
    which holds for $1 \geq z \geq 0$ and $1\geq \varepsilon \geq 0$; we provide a proof of inequality \eqref{eq:appe:nontrivial_1} in Appendix \ref{claim:A}.   
    Thus, we rewrite equation \eqref{eq:appe:1} as 
    \begin{equation}\label{eq:appe:2}
        \begin{split}
            w(L+1,y) \leq  w(L,y) -  \varepsilon_1 w(L,y) m(x_L, y).
        \end{split}
    \end{equation}
    We sum equation \eqref{eq:appe:2} for over all $y\in {\cal Y}$ and we get
    \begin{equation}\label{eq:appe:3}
        \begin{split}
            &\sum_{y\in{\cal Y}}w(L+1,y) \\&\quad\leq  \sum_{y\in{\cal Y}}w(L,y) 
            - \varepsilon_1\sum_{y\in {\cal Y}}  w(L,y) m(x_L,y)
            \\&\quad= \bigg(\sum_{y\in {\cal Y}}w(L,y)\bigg) \bigg(1- \varepsilon_1\sum_{y\in{\cal Y}}f(L,y) m(x_L,y) \bigg)
            \\&\quad \leq \bigg(\sum_{y\in {\cal Y}}w(L,y)\bigg) \exp \bigg(- \varepsilon_1\sum_{y\in{\cal Y}}f(L,y) m(x_L,y) \bigg),
        \end{split}
    \end{equation}
    where the second equality follows from $$f(l,y) = \frac{w(l,y)}{\sum_{y\in {\cal Y}}w(l,y)},$$ and the last inequality follows from $1+z \leq e^z$ for any real number $z$. Calculate equation \eqref{eq:appe:3} recursively, the upper bound is 
    \begin{equation} \label{eq:appe:upperbound}
        \begin{split}
            &\sum_{y\in{\cal Y}} w(L+1,y) 
            \\& \quad \leq \bigg(\sum_{y\in {\cal Y}}w(L,y) \bigg)\exp \bigg(-\varepsilon_1 \sum_{y\in{\cal Y}}f(L,y) m(x_L,y) \bigg)
            \\& \quad\leq \bigg(\sum_{y\in {\cal Y}}w(L-1,y) \bigg)
            \\ & \qquad \quad \exp \bigg(-\varepsilon_1 \sum_{l=L-1}^L \sum_{y\in{\cal Y}} f(l,y) m(x_l,y) \bigg)
            \\&\quad \leq   \cdots
            \\& \quad\leq \bigg(\sum_{y\in {\cal Y}}w(1,y) \bigg)\exp \bigg(-\varepsilon_1 \sum_{l=1}^L\sum_{y\in{\cal Y}} f(l,y) m(x_l,y) \bigg)
            \\&\quad = |{\cal Y}| \exp \bigg(-\varepsilon_1 \sum_{l=1}^L\sum_{y\in{\cal Y}} f(l,y) m(x_l,y) \bigg),
        \end{split}
    \end{equation}
    where the last equality follows from $\sum_{y\in{\cal Y}} w(1,y) = \sum_{y\in{\cal Y}}1 = |{\cal Y}|.$
    On the other hand, the lower bound is 
    \begin{align}\label{eq:appe:lowerbound}
            \sum_{y\in {\cal Y}} w(L+1,y) &= \sum_{y\in {\cal Y}} \exp{\bigg(-\varepsilon \sum_{l=1}^L m(x_l, y)\bigg)} \notag
            \\ & \geq \exp{\bigg(-\varepsilon \sum_{l=1}^L m(x_l, y)\bigg)},
    \end{align}
    for every $y\in {\cal Y}$ where the inequality follows that $e^z$ is positive for any real number $z$. Combining the upper bound of equation \eqref{eq:appe:upperbound} and the lower bound of equation \eqref{eq:appe:lowerbound}, we get
    \begin{equation*}
        \begin{split}
            &\exp\bigg(-\varepsilon \sum_{l=1}^L m(x_l, y) \bigg) 
            \\& \quad \leq 
        |{\cal Y}| \exp \bigg( -\varepsilon_1 \sum_{l=1}^L \sum_{y'\in{\cal Y}}f(l,y)m(x_l, y')\bigg),
        \end{split} 
    \end{equation*}
    for every $y\in {\cal Y}$. Simplify the above equation with
    \begin{equation}
    \label{eq:appe:nontrivial_2}
        \varepsilon_1 = 1 - e^{-\varepsilon} \geq \varepsilon(1-\varepsilon),
    \end{equation}
    which holds for $\frac{1}{2}\geq \varepsilon \geq 0$; we provide a proof of inequality \eqref{eq:appe:nontrivial_2} in Appendix \ref{claim:B}. And taking the natural logarithm and dividing by $-\varepsilon$, we get
    \begin{equation*}
        \begin{split}
            &\sum_{l = 1}^L m(x_l,y) 
            \\ & \qquad \geq (1-\varepsilon) \sum_{l=1}^L\bigg(\sum_{y'\in {\cal Y}} f(l,y)m(x_l,y') \bigg) - \frac{\ln|{\cal Y}|}{\varepsilon}.
        \end{split}
    \end{equation*}
\\\qed

\subsubsection{Proof of nontrivial inequality \eqref{eq:appe:nontrivial_1}}
    \label{claim:A}
    We define 
    $$f(\varepsilon) := \exp{(-\varepsilon z)} - (1 - (1-\exp{(-\varepsilon)}) z).$$
    Taking the derivatives of $f(\varepsilon)$ with respect to $\varepsilon$, we get
    $$f'(\varepsilon) = - z\exp{(-\varepsilon z)} + z\exp{(-\varepsilon)} \leq 0,$$
    where the inequality follows from $\exp{(-\varepsilon z)} \geq \exp{(-\varepsilon)}$ for $1\geq z \geq 0.$
    We get
    $$f(\varepsilon)\leq f(0) = 0.$$
    Thus, we get the inequality \eqref{eq:appe:nontrivial_1} such as
    $$\exp{(-\varepsilon z)} \leq  1 - (1-\exp{(-\varepsilon)}) z.$$
\qed

\subsubsection{Proof of nontrivial inequality \eqref{eq:appe:nontrivial_2}}
\label{claim:B}
We define 
$$f(\varepsilon) := 1- \exp{(-\varepsilon)} - \varepsilon(1-\varepsilon).$$
Taking the first and second derivatives of $f(\varepsilon)$ with respect to $\varepsilon$, we get
$$f'(\varepsilon) = \exp{(-\varepsilon)} -1 + 2\varepsilon,$$
and
$$f''(\varepsilon) = -\exp{(-\varepsilon)} + 2.$$
Because $f''(\varepsilon)>0$ for $\frac{1}{2}\geq \varepsilon \geq 0$, we get 
$$f'(\varepsilon) \geq f'(0) = 0.$$
Then, we get
$$f(\varepsilon)\geq f(0) = 0.$$
We get the inequality \eqref{eq:appe:nontrivial_2} such as
$$ 1 - e^{-\varepsilon} \geq \varepsilon(1-\varepsilon).$$
\qed

\subsection{Proof of Proposition \ref{propo:sourcdResolv}}
    \label{Proof:Proposition}
    We consider type $T\in {\cal T}$ and define a set
    \begin{equation}
    \label{eq:appe:set_D}
        {\cal D} := \bigg\{T \in {\cal T} : P_{\cal T}(T) \geq \frac{\tau'}   {|{\cal T}|} \bigg\}.
    \end{equation}
    We shall apply Theorem \ref{the:Single-shot} for type distribution $P_{\cal T}$ as an input distribution $P$, a noiseless channel $\mathbf{1}[T'=T]$ as a channel $W$, and ${\cal D}\times {\cal D}$ as arbitrary set ${\cal S}$ respectively.

    First, we calculate what corresponds to $D_\mathrm{max}$. 
    A truncated measure with noiseless channel is $\mathbf{1}[T'=T]\; \mathbf{1}[(T,T')\in({\cal D\times D})]$. 
    The truncated measure with type distribution $P_{\cal T}$ satisfies that
    \begin{align*}
        P_{\cal T}^{\cal D}(T')&= \sum_{T\in{\cal T}}P_{\cal T}(T) \mathbf{1}[T'=T] \mathbf{1}[(T,T')\in({\cal D\times D})]
        \\& = P_{\cal T}(T') \mathbf{1}[T'\in{\cal D}].
    \end{align*}
    The set ${\cal D}$ is equivalent to support set $\mathsf{s}(P_{\cal T}^{\cal D})$.
    Thus, we get
    \begin{equation}\label{eq:prop:D_max}
        \begin{split}
            D_\mathrm{max} &= \max_{(T,T')\in ({\cal D\times D})} \ln \frac{\mathbf{1} [T'=T] \; \mathbf{1}[(T,T')\in({\cal D\times D})]}{P_{\cal T}^{\cal D}(T')}
            \\&= \max_{T\in{\cal D}} \ln \frac{1}{P_{\cal T}^{\cal D}(T)}
            \\&\leq \ln{\frac{|{\cal T}|}{\tau'}}.
        \end{split}
    \end{equation}
    where the inequality follows from definition \eqref{eq:appe:set_D}.
    
    Next, we calculate sum of complements and get
    \begin{align}\label{eq:prop:smoothing}
        \sum_{(T,T')\in({\cal D}^c, {\cal D}^c)} P_{\cal T}(T)\mathbf{1}[T'=T] &= \sum_{T\in{{\cal D}^c}} P_{\cal T}(T) \notag \\
        &\leq \sum_{T\in{{\cal D}^c}} \frac{\tau'}{|{\cal T}|} \notag \\
        &\leq \tau'
    \end{align}
    where the first inequality follows from definition \eqref{eq:appe:set_D}.
    Applying Theorem \ref{the:Single-shot} for equation \eqref{eq:prop:D_max} and \eqref{eq:prop:smoothing}, the codebook ${\cal C'} =\{T_1, T_2, ..., T_{L'} \}\subset {\cal T}$ obtained by the MWU algorithm makes a distribution $R_{{\cal C}'}(T)$ such as
    $$R_{{\cal C}'}(T) = \frac{1}{L'} \sum_{l=1}^{L'} \mathbf{1}[T_l = T].$$
    The distribution $R_{{\cal C}'}(T)$ satisfies 
    \begin{equation}\label{eq:prop:last}
        d_\mathrm{var}(P_{\cal T}, R_{{\cal C}'})\leq 2\varepsilon' + \tau',
    \end{equation}
    as long as the size $L'$ satisfies 
    $$L' \geq \frac{|{\cal T}| \ln{|{\cal T}|}}{\tau' \varepsilon'^2}.$$
    Note that the total number of types is at most $(n+1)^{|{\cal X}|}$. Thus, equation \eqref{eq:prop:last} satisfies as long as the size $L'$ satisfies equation \eqref{eq:L_n'}.\\
\qed

\subsection{Proof of Remark \ref{Remark_}}
\label{proof:Remark}
We consider i.i.d. input distribution $P_n=P^n$ in Theorem \ref{the:Main_result}. Because of the type definition, the number of sequences $|{\cal X}_T^n|$ in a type satisfies that
\begin{equation}
    \label{eq:number_of_sequences_in_type}
    |{\cal X}_T^n|\leq \exp{(nH(T))}.
\end{equation}
Consider i.i.d. distribution $P_n=P^n$. The probability of $P_n(x^n)$ for $x^n\in{\cal X}_T^n$ satisfies that
\begin{align}
\label{eq:i_i_d}
    P_n(x^n) &= \prod_{a\in{\cal X}} P(a)^{N(a|x^n)} \notag
    \\&= \exp{\bigg(\sum_{a\in{\cal X}} n \bigg(\frac{1}{n}N(a|x^n)\ln{P(a)}  \bigg) \bigg)}\notag
    \\&= \exp{\bigg(n\sum_{a\in{\cal X}} T(a)\ln{P(a)} \bigg)}\notag
    \\&= \exp{\Big(-n\big(H(T)+D(T||P)\big) \Big)}
\end{align}
where $N(a|x^n)$ outputs the number of occurrences of $a$ in $x^n$, the second last equality follows that $\frac{1}{n}N(a|x^n)=T(a)$, and $D(T||P)$ is KL-divergence. From the definition of type distribution \eqref{eq:definition_of_type_distribution}, combining equation \eqref{eq:number_of_sequences_in_type} and equation \eqref{eq:i_i_d}, we get
\begin{equation}
    \label{eq:P_n(X_T^n)}
    P_{\cal T}(T) = P_n({\cal X}_T^n) \leq \exp{(-nD(T||P))}.
\end{equation}

Consider to apply equation \eqref{eq:P_n(X_T^n)} to 
$$\sum_{T\in{\cal T}}P_{\cal T}(T) d_\mathrm{var}(W_{P_n^T},W_{{\cal C}_n^T}),$$
in equation \eqref{eq:sec6:d_var_}.
For an arbitrary small positive $\nu$, we divide all types into types which satisfies that $D(T||P) > \nu$ and the other types. For the types which satisfies that $D(T||P) > \nu$, because the total number of types is at most $(n+1)^{|{\cal X}|}$, we get $$(n+1)^{|{\cal X}|}\exp{(-nD(T||P))} < (n+1)^ {|{\cal X}|} \exp{(-n \nu)} \;\leq \varepsilon, $$
for sufficiently large $n$. Thus applying this relation to equation \eqref{eq:sec6:d_var_}, we get
\begin{align*}
    d_\mathrm{var}(W_{P},W_{\cal C}) &\leq \sum_{T\in {\cal T}:\atop D(T||P)\leq \nu}P_{\cal T}(T) d_\mathrm{var}(W_{P_n^T},W_{{\cal C}_n^T}) \\&\quad+ (n+1)^ {|{\cal X}|} \exp{(-n \nu)} + d_\mathrm{var}(P_{\cal T},R_{\cal C'}).
\end{align*}
On the other hand, for $D(T||P)\leq \nu$, Pinsker's inequality and data processing inequality imply that
\begin{align}
    \label{eq:pinsker's}
    d_\mathrm{var}(W_P,W_T)\leq d_\mathrm{var}(P,T) \leq \sqrt{\frac{1}{2}\nu}.
\end{align}
For $x^*=\underset{x}{\mathrm{argmax}}\;H(W_x)$, we get
\begin{align}
    \label{eq:Muture_1}
    |H(W|P)-H(W|T)| &\leq \sum_{x\in\mathcal{X}} |P(x)-T(x)| H(W_x) \notag
    \\& \leq  2\bigg(\frac{1}{2}\sum_{x\in\mathcal{X}} |P(x)-T(x)| H(W_{x^*})\bigg) \notag
    \\&\leq \sqrt{2\nu} H(W_{x^*}) \notag
    \\&\leq \sqrt{2\nu} \ln{|\mathcal{Y}|},
\end{align}
where the third inequality follows from equation \eqref{eq:pinsker's} and the last inequality follows from the upper bound of entropy. Continuity of entropy in \cite[Lemma 2.7]{Csiszar-Korner11} implies that for $d_\mathrm{var}(W_P,W_T)\leq \frac{1}{2}$,
\begin{align}
    \label{eq:Muture_2}
    |H(W_P)-H(W_T)|&\leq -d_\mathrm{var}(W_P,W_T) \ln{\frac{d_\mathrm{var}(W_P,W_T)}{|\mathcal{Y}|}} \notag
    \\ &\leq \sqrt{\frac{1}{2}\nu} \ln{{|\mathcal{Y}|}} - \sqrt{\frac{1}{2}\nu} \ln{\sqrt{\frac{1}{2}\nu}},
\end{align}
where the last inequality follows that $-t \ln \frac{t}{|{\cal Y}|}$ is a monotonically non-decreasing function for $0\leq t\leq \frac{1}{2}$ with equation \eqref{eq:pinsker's}.
From equation \eqref{eq:Muture_1} and \eqref{eq:Muture_2}, we get 
\begin{align*}
    &|I(T,W) - I(P,W)|
    \\&\quad\leq |H(W_P)-H(W_T)| +|H(W|P)-H(W|T)|
    \\&\quad \leq \frac{3\sqrt{2\nu}}{2} \ln{{|\mathcal{Y}|}} - \sqrt{\frac{1}{2}\nu} \ln{\sqrt{\frac{1}{2}\nu}}.
\end{align*}
Considering the types where $I(T,W)$ approximates $I(P,W)$ in Theorem \ref{the:chaResol_fixedType}, we get $$d_\mathrm{var}(W_{P_n^T},W_{{\cal C}_n^T})\leq 2\varepsilon + \tau +\delta,$$ as long as the size $L_n$ satisfies that 
\begin{align*}
    &L_n \geq \;\exp{\bigg(n\bigg(I(P,W) +\frac{3\sqrt{2\nu}}{2} \ln{|\mathcal{Y}|} }
    \\& \qquad\quad \qquad - \sqrt{\frac{1}{2}\nu} \ln {\sqrt{\frac{1}{2}\nu}} + \alpha'+\beta'\bigg) -\ln{\tau}\bigg) 
    \times\frac{\ln{|{\cal Y}|^n} }{\varepsilon^2}.
\end{align*}
where the inequality follows from $I(T,W)\leq I(P,W)+\frac{3\sqrt{2\nu}}{2} \ln{|\mathcal{Y}|} - \sqrt{\frac{1}{2}\nu} \ln {\sqrt{\frac{1}{2}\nu}}$.
In the same flow as Theorem \ref{the:Main_result}, 
by setting $\varepsilon,\varepsilon', \tau, \tau', \delta,\alpha', \beta', \nu$ to small values, for sufficiently large $n$, the sum of small values is less than $\kappa$.
Thus, we get the codebook ${\cal C}_n$ with size $L_n''$ which only satisfies
$$L_n'' \geq \exp\Big(n\big(I(P,W)+\kappa\big)\Big).$$
\qed

\bibliographystyle{plain}
\bibliography{sample}

\end{document}